\documentclass{article}

\usepackage[utf8]{inputenc}
\usepackage{amsthm}
\usepackage{amsfonts}
\usepackage[abbrev]{amsrefs}		
\usepackage{hyperref}
\usepackage{geometry}
\usepackage{authblk}
\usepackage{tikz}

\usepackage{multirow}
\usepackage{amsmath,amsfonts}
\usepackage{amssymb,amsthm}
\usepackage{mathrsfs,stmaryrd}
\usepackage[all,cmtip]{xy}
\usepackage{tensor}
\usepackage[makeroom]{cancel}
\usepackage{enumitem}

\theoremstyle{plain}
\newtheorem{defi}{Definition}[section]

\newtheorem{thm}[defi]{Theorem}

\theoremstyle{remark}

\newtheoremstyle{mystyle}
  {}
  {}
  {\itshape}
  {}
  {\bfseries}
  {.}
  { }
  {}
\theoremstyle{mystyle}

\newcounter{note}
\DeclareMathAlphabet{\mathpzc}{OT1}{pzc}{m}{it}

\def\reR{\mathbb{R}}

\def\p{\partial}

\def\<{\langle}
\def\>{\rangle}

\def\and{\text{ and }}

\usepackage{xy}
\xyoption{matrix}
\xyoption{frame}
\xyoption{arrow}
\xyoption{arc}

\usepackage{ifpdf}
\ifpdf
\else
\PackageWarningNoLine{Qcircuit}{Qcircuit is loading in Postscript mode.  The Xy-pic options ps and dvips will be loaded.  If you wish to use other Postscript drivers for Xy-pic, you must modify the code in Qcircuit.tex}
\xyoption{ps}
\xyoption{dvips}
\fi

\entrymodifiers={!C\entrybox}

\newcommand{\bra}[1]{{\left\langle{#1}\right|}}
\newcommand{\ket}[1]{{\left|{#1}\right\rangle}}
\newcommand{\qw}[1][-1]{\ar @{-} [0,#1]}
\newcommand{\qwx}[1][-1]{\ar @{-} [#1,0]}

\newcommand{\control}{*!<0em,.025em>-=-<.2em>{\bullet}}

\newcommand{\ctrl}[1]{\control \qwx[#1] \qw}

\newcommand{\targ}{*+<.02em,.02em>{\xy ="i","i"-<.39em,0em>;"i"+<.39em,0em> **\dir{-}, "i"-<0em,.39em>;"i"+<0em,.39em> **\dir{-},"i"*\xycircle<.4em>{} \endxy} \qw}
\newcommand{\Qcircuit}{\xymatrix @*=<0em>}

\def\p{\partial}
\def\reR{\mathbb{R}}

\begin{document}

\title{A Data Driven Approach to Learning The Hamiltonian Matrix in Quantum Mechanics}
\author[1]{Jordan Burns \thanks{eamil: jordan.burns@usu.edu}} 
\author[2]{David Maughan \thanks{email: drm@rincon.com}} 
\author[1]{Yih Sung \thanks{email: yih.sung@usu.edu}}
\affil[1]{Department of Mathematics, Utah State University}
\affil[2]{Rincon Research Corporation}

\date{}

\maketitle

\begin{abstract}
We present a new machine learning technique which calculates a real-valued, time independent, finite dimensional Hamiltonian matrix from only experimental data. A novel cost function is given along with a proof that the cost function has the theoretically correct Hamiltonian as a global minimum. We present results based on data simulated on a classical computer and results based on simulations of quantum systems on IBM's ibmqx2 quantum computer. We conclude with a discussion on the limitations of this data driven framework, as well as several possible extensions of this work. We also note that algorithm presented in this article not only serves as an example of using domain knowledge to design a machine learning framework, but also as an example of using domain knowledge to improve the speed of such algorithm. 

\end{abstract}

\section{Introduction}
\emph{"So the problem is: Know your Hamiltonian!" - Richard Feynman}

\smallskip

In quantum mechanics, the Hamiltonian is a mathematical operator which completely determines the evolution of a quantum mechanical system.  If one knows the Hamiltonian governing a system, then any question of what will happen to the system can be answered \cite{book_feynman_lec_3}*{Vol.3~Ch.8}. Traditionally, finding the Hamiltonian involves formidable calculations, as well as deep theoretical knowledge of the system under study \cite{book_feynman_lec_3}*{Vol.3~Ch.9, Ch.12}. However, when the quantum system is too complicated, it is almost impossible to find the explicit expression of the Hamiltonian by hands. In contrast to the classical approach, we use a novel data driven machine learning algorithm to discover the Hamiltonian. Thus, the method we present in this paper involves only a knowledge of the basis states of the system. Towards the end of the paper we also discuss how classical approaches and the data driven approach can be blended to produce more accurate results.

Our goal and method are different from the quantum Hamiltonian Learning (QHL) introduced in \cite{ref_Hamiltonian_learning_14, ref_Hamiltonian_learning_imperfect_14}. QHL is aiming to simulate an existing quantum system by Bayesian updating. In QHL, the probability function  $P(D | H)=|\bra{D}e^{iHt}\ket{\psi}|^2$ is used as a data point to update the probability $P(H | D)$; however, in our approach, our goal is to \emph{build} a quantum system to solve a peculiar question. Driven by data, we recover the parameters $\{w_{ij}\}$ of the Hamiltonian $H$ by optimizing the cost function $C(w_{ij})=1-|\bra{D}e^{iH(w_{ij})t}\ket{\psi}|^2$. Suggested by the nature of the method, we name it Data Driven Approach (DDA).

Quantum computing is usually seen as a way to speed up statistical computation to improve machine learning algorithm. Quantum Machine Learning (QML) is a subfield particularly focusing on performing machine learning algorithm on a quantum computer \cite{ref_Lloyd13:quantumalgorithms,ref_biamonte:quan_machine_learning}. In this paper,  we apply in the opposite direction: we use machine learning techniques to search and find the right Hamiltonian matrix in quantum computing. We have succeeded in recovering the Hamiltonian of Grover's search Algorithm, Ammonia Molecule's quantum system, Hyperfine Splitting of Hydrogen, and random Hamiltonian by DDA.

We structure our article as follows: in Section \ref{sec:methods} we explain the setting and outline the structure of gradient decent approach. In Section \ref{sec:results}, we use several examples to demonstrate the power of the new method. In Section \ref{sec:conclusion}, we draw our conclusion. In Appendix \ref{app:min_H_is_not_uniq}, we show that a minimum Hamiltonian is not unique. It is allowed to have a global phase difference. 


\section{Methods} \label{sec:methods}
\subsection{Data Driven Approach (DDA)}
\subsubsection*{The Schrodinger Equation And Setting}

Given a quantum dynamical system, the evolution of the system is characterized by the Sch\"odinger equation
\begin{equation} \label{eq:quan_evolution_simplified}
    i\hbar \frac{d}{dt} \ket{\psi(t)} = H \ket{\psi(0)},
\end{equation}
where $\psi(t)$ represents the quantum states at time $t$ and $H$ is the Hamiltonian operator (cf. \cite{book_quan_comp&quan_info}*{Ch 2}). When $H$ is independent of time the differential equation has the following solution (at least upto some $i$'s and $\hbar$'s)

\begin{equation} \label{eq:quan_evolution}
    \ket{\psi(t)} = e^{-itH} \ket{\psi(0)},
\end{equation}
where $\ket{\psi(0)}$ is the initial state. By absorbing $t$ and $\hslash$ into the Hamiltonian $H$ and writing
$$
\hat H = i H,
$$
we can treat (\ref{eq:quan_evolution}) as a equation connects initial state and final state
\begin{equation} \label{eq:quan_evolution_simplified}
\ket{\psi_{\text{final}}} = e^{-i\hat H} \ket{\psi_{\text{init}}}.
\end{equation}
Therefore, the task becomes finding an appropriate Hamiltonian $\hat H$ so that the given potential initial states can be mapped to the desired final states. In general, the matrix consisting of $\{\ket{\psi_{init}}\}$ may not be unitary, so (\ref{eq:quan_evolution_simplified}) is \emph{not} a simple inverse matrix problem. Therefore, we need a new approach.

In the case that the quantum system of interest can only exist in superpositions of finitely many basis states (which is the only case we consider in this paper), then the Hamiltonian operator can be expressed as a matrix. In quantum computing, this is exactly the case. With finitely many input states and output sates, the computation can be realized by a quantum circuit. In general a Hamiltonian is a hermitian operator. In order to demonstrate our method transparently, we restrict ourselves to only \emph{real-valued} Hamiltonians, so for this paper the Hamiltonian is a symmetric matrix. Nevertheless, it is very easy to generalize the whole scheme over \emph{complex} numbers.


From a high level point of view, the data driven method is to initialize a random Hamiltonian, which we call the learning Hamiltonian. Using an optimization technique, the entries of the learning Hamiltonian are slowly adjusted until the evolution of the learning Hamiltonian coincides with the data measured in the laboratory. This process is monitored by the means of a cost function. The cost function can be interpreted as the likelihood that the data was produced by the learning Hamiltonian. We now present background information needed to examine the method in further detail.

\subsubsection*{Parameterizing the space of all Hamiltonian matrices}
Consider a quantum system which can exist in arbitrary superpositions of $n$ basis states. The evolution of this system will be governed by an $n$ by $n$ Hamiltonian matrix. Since the Hamiltonian is a symmetric matrix, it requires $n(n+1)/2$ numbers to specify a Hamiltonian matrix. We simply allow the entries $w_{ij}$ of the Hamiltonian $H$ to be the parameters for tuning. Hence, the sample space of all possible Hamiltonians can be characterized as all possible combinations of 
$$
\{w_{ij}\}_{1\le i\le j\le n}.
$$
To align ourselves with usual machine learning language we shall call these parameters weights. For the sake a clarity we present our parameterization of the space of $4\times 4$ Hamiltonian matrices 

\[
\begin{bmatrix}
    w_{11}  & w_{12} & w_{13} & w_{14} \\
    w_{12}  & w_{22} & w_{23} & w_{24} \\
    w_{13}  & w_{23} & w_{33} & w_{34} \\
    w_{14}  & w_{24} & w_{34} & w_{44} \\
\end{bmatrix},
\]
where each $w_{ij}$ is a real number.


\subsubsection*{Choosing a basis}
Since our parameterized Hamiltonian is still arbitrary, it has not presupposed a basis of the Hilbert Space of states. One convenient basis for the states of a quantum system is the standard basis of $R^n$, where we identify each basis element with exactly one measurable state of the system. That is to say, a basis state is a state of the system which collapses to the associated measurable state with probability one. For example, suppose we had a two-state spin system. The two measurable states are spin up and spin down. The spin up state could correlate to $(1,0)^T$ and the spin down state could correlate to $(0,1)^T$. This choice is not unique, but is well-suited to our task.


\subsubsection*{Defining the data}
For this approach, a set of input states and output states are needed to effectively learn the desired Hamiltonian. By preparing a known input state $\ket{\psi_i}$ and measuring the state of the quantum system after a specific time $t$, an input and output pair can be found. Due to the randomness of measured states in quantum mechanics, this process can (and should) be repeated to obtain the expected value for the measured output state. Call this expected value $\ket{\phi_i}$.
In the theoretical limit of an infinite number of repetitions of preparing and measuring data

\begin{equation*}
    \ket{\phi_i} = \ket{\psi_i(t)} = e^{-itH} \ket{\psi(0)}
\end{equation*}
and
\begin{equation*}
    \ket{\psi_i} = \ket{\psi_i(0)}
\end{equation*}

These sets of data points, $\ket{\psi_i}$, $\ket{\phi_i}$, and $t$, provide the foundation of the cost function that is central to this approach.


\subsection{The Cost Function}
As previously mentioned, the cost function is  the probability of the learning Hamiltonian giving the incorrect evolution. It is characterized as a mapping from the parameterized space of Hamiltonians to closed interval of the real line $[0,1]$ defined by
\begin{equation} \label{eq:cost_func}
    C(H_L) = \sum_{i=1}^m 1 - |\bra{\phi_i} e^{-i\hbar H_L t_i} \ket{\psi_i}|^2
\end{equation}
where $H_L=H_L(w_{ij})$ is a learning Hamiltonian in the parameterized space, $m$ is the total number of data pairs, $\psi$ and $\phi$ come from the data pairs, and $t_i$ is the amount of time for which data pair $i$ evolved. Let $C_T$ be the true Hamiltonian that governed the quantum evolution. As $C(H_L)\ge 0$ and $C(C_T)=0$, namely,  its associated cost is zero, $C_T$ is a global minimum.

Since we have parameterized the space of Hamiltonians by the weights $\{w_{ij}\}$, we now have well defined partial derivatives of the cost function
$$
\frac{\p C(H_L)}{\p w_{ij}}
$$
with respect to these parameters $\{w_{ij}\}$ of $H_L$. These partial derivatives allow us for the use of a large family of optimization techniques. 


\subsubsection*{Finding the global minimum}
We have already shown that the true Hamiltonian is a global minimum of the cost function and that the cost is zero at this point. Potential solvers include stochastic gradient descent, simultaneous perturbation stochastic approximation, or another method that is not a hill climbing technique. For simple cases, such as $2\times 2$ Hamiltonians, we can derive a general analytic formula for the partial derivatives to boost the speed of convergence in optimization. For larger systems we implemented a momentum based gradient descent which used numerical finite differences to approximate the partial derivatives when the system is too complicated to obtain an analytic solution.

The formulas for momentum based gradient descent are

\begin{gather}
    V_{t} = \beta V_{t-1,ij} + \alpha \nabla C(W_{t-1}) \\
    W_t = W_{t-1} - V_{t},
\end{gather}
where $\beta$ is a decay parameter between 0 and 1 and $\alpha$ is the learning rate. Note that if $\beta = 0$, we recover the regular gradient descent.

Note that as we are minimizing the cost function and the cost function is invariant by a global phase, there is no unique Hamiltonian in optimizing the cost function. In other words, 
$$
C(H_L) = C(H_L + f\cdot I),
$$
where $f\in\reR$ and $I$ is the identity matrix. We include the proof in Appendix \ref{app:min_H_is_not_uniq}. 


\subsubsection*{Accuracy Estimate}
Recall the maximum matrix norm of a matrix $A$ is defined as 
\begin{equation}
\|A\| := \sup\{ |Ax| \mid x\in \reR^n \,\text{ with }\, \|x\|=1 \}.
\end{equation}
Thus we can measure how accurate our approximation is in uncovering random Hamiltonians by
$$
\|H_{\text{rand}} - H_{\text{learn}}\|,
$$
where $H_{\text{rand}}$ is the random Hamiltonian and $H_{\text{learn}}$ is the Hamiltonian by learning.


\section{Results} \label{sec:results}



\begin{table}
\begin{tabular}{l|l|l|l|l}
\hline
\multicolumn{5}{|c|}{IBM Q Generated Data}  \\ \hline
& \multicolumn{4}{l}{Number of times each basis state was measured after evolution} \\ \hline
Prepared initial State & State 1             & State 2            & State 3            & State 4            \\ \hline
State 1 & 993 & 12 & 8 & 11\\   \hline
State 2 & 34 & 13 & 959 & 18\\  \hline
State 3 & 17 & 982 & 14 & 11\\  \hline
State 4 & 10 & 34 &  47 & 933\\ \hline
Uniform Superposition & 240 & 249 & 255 & 280 \\ 
\hline
\end{tabular}
\caption{\small Finding The Hamiltonian of The Hyperfine Splitting of Hydrogen} \label{tab:hamiltonian_HSH}
\end{table}

\subsection{The Hyperfine Splitting of Hydrogen: A thoroughly worked example}
In this section, we present how DDA discovers a Hamiltonian in a physical quantum system in detail. 

Hyperfine structure is the shifts and splits in the energy levels of atoms and molecules that is caused by interactions of the nucleus (or nuclei, in molecules) with internally generated electromagnetic fields. In the hydrogen atom, the hyperfine splitting comes from the spin states of the proton and electron. The Hamiltonian governing the hyperfine splitting of hydrogen is extensively studied in the Feynman lectures \cite{book_feynman_lec_3}*{Vol.3~Ch.12}, in which he explains how the spherical symmetry of the problem forces the Hamiltonian to have a certain form. Here, we will now show how DDA can find the same Hamiltonian, but without any knowledge besides knowing the basis states of the system. 

Suppose the measurable states of the system are the electron and proton both with spin up, the electron spin up and the proton spin down, the electron spin down and the proton spin up, and both particles spin down. We label these basis vectors as state one, state two, state three, and state four, respectively. We concisely label these states as $\ket{\uparrow\uparrow}, \ket{\uparrow\downarrow}, \ket{\downarrow\uparrow}, \ket{\downarrow\downarrow}$ respectively. We then identify each of the basis states with an element of the standard basis of $R^4$ as follows.
\[
\ket{\uparrow\uparrow} = 
\begin{bmatrix}
    1 \\ 0 \\ 0 \\ 0 \\
\end{bmatrix},
\ket{\uparrow\downarrow} = 
\begin{bmatrix}
    0 \\ 1 \\ 0 \\ 0 \\
\end{bmatrix},
\ket{\downarrow\uparrow} = 
\begin{bmatrix}
    0 \\ 0 \\ 1 \\ 0 \\
\end{bmatrix},
\ket{\downarrow\downarrow} = 
\begin{bmatrix}
    0 \\ 0 \\ 0 \\ 1 \\
\end{bmatrix}
\]

Suppose we have a way to measure in which of these states the atom is. We could then prepare a specific state; allow it to evolve for some amount of time, and then measure which state the system is in. Repeating this analysis many times for the same state and time would allow us to develop the probabilities of the initial state collapsing to different basis states after the evolution. These probabilities are then organized into a vector and become the second element in our data pairs. 

In a sense we can simulate the hyperfine splitting of hydrogen on a quantum computer using only a SWAP gate: 
$$
\Qcircuit @C=1em @R=.7em {
 & \ctrl{1} & \targ & \ctrl{1} & \qw & \\ 
 & \targ & \ctrl{-1} & \targ & \qw & 
}
\quad = \quad
\begin{bmatrix}
1 & 0 & 0 & 0 \\
0 & 0 & 1 & 0 \\
0 & 1 & 0 & 0 \\
0 & 0 & 0 & 1 
\end{bmatrix}.
$$
Table \ref{tab:hamiltonian_HSH} summarizes our data. The data simulates an evolution lasting 0.785 seconds and then being measured. Each of our initial states was prepared, evolved, and then measured 1024 times.

From the table we can build our data pairs. They are the column vectors as follows. 

\begin{gather}
\notag \left(
\begin{bmatrix}
1 \\ 0 \\ 0 \\ 0\\
\end{bmatrix}
,
\begin{bmatrix}
\sqrt{\frac{993}{1024}} \\ \sqrt{\frac{12}{1024}} \\ \sqrt{\frac{8}{1024}} \\ \sqrt{\frac{11}{1024}}\\
\end{bmatrix}
\right)
,
\left(
\begin{bmatrix}
0 \\ 1 \\ 0 \\ 0\\
\end{bmatrix}
, 
\begin{bmatrix}
\sqrt{\frac{34}{1024}} \\ \sqrt{\frac{13}{1024}} \\ \sqrt{\frac{959}{1024}} \\ \sqrt{\frac{18}{1024}}\\
\end{bmatrix}
\right)
,
\left(
\begin{bmatrix}
0 \\ 0 \\ 1 \\ 0\\
\end{bmatrix}
,
\begin{bmatrix}
\sqrt{\frac{17}{1024}} \\ \sqrt{\frac{982}{1024}} \\ \sqrt{\frac{14}{1024}} \\ \sqrt{\frac{11}{1024}}\\
\end{bmatrix}
\right)
, \\
\left(
\begin{bmatrix}
0 \\ 0 \\ 0 \\ 1\\
\end{bmatrix}
,
\begin{bmatrix}
\sqrt{\frac{4}{1024}} \\ \sqrt{\frac{10}{1024}} \\ \sqrt{\frac{34}{1024}} \\ \sqrt{\frac{933}{1024}}\\
\end{bmatrix}
\right)
,
\left(
\begin{bmatrix}
\sqrt{.25} \\ \sqrt{.25} \\ \sqrt{.25} \\ \sqrt{.25}\\
\end{bmatrix}
,
\begin{bmatrix}
\sqrt{\frac{240}{1024}} \\ \sqrt{\frac{249}{1024}} \\ \sqrt{\frac{255}{1024}} \\ \sqrt{\frac{280}{1024}}\\
\end{bmatrix}
\right)
\end{gather}

The algorithm gives this as the final output where the cost is 0.0446286.
\[
H_L = 
\begin{bmatrix}
1.2186 & 0.03709 & -0.03641 & 0.00055 \\
0.03709 & -0.79105 & 1.99915 & -0.00653 \\
-0.03641 & 1.99915 & -0.75919 & 0.00611 \\
0.00055 & -0.006536 & 0.00611 & 1.22186 \\
\end{bmatrix}
\]
This is greatly close to the actual Hamiltonian 
$$
H_{\text{real}} = 
\begin{bmatrix}
1 & 0 & 0 & 0 \\
0 & -1 & 2 & 0 \\
0 & 2 & -1 & 0 \\
0 & 0 & 0 & 1
\end{bmatrix}
$$
considering we can adjust the diagonals according to Theorem \ref{thm:Hamiltonaina_is_not_uniq}.

The true Hamiltonian $H_{\text{real}}$ results in a cost of  0.044891608584168144. This indicates that our method found a Hamiltonian that is very accurate considering the noise present in the system (the data is 5\% noisy).


\subsection{Numerical Experiments}

Successfully studying the evolution of the hyperfine splitting of hydrogen using real quantum data was a great success. To be sure that the method would work on a large variety of systems and that it would scale with size, we also generated random Hamiltonians and studied them with our data driven approach. 


\subsubsection*{The data}



There are uncountably many different ways in which one could prepare input data. One is welcome to prepare any arbitrary superpositions and measure the evolution (and thereby collapse the wavefunction) at any time. If $n$ is the number of states then we decided to use $1 + n + {n \choose 2}$ input  data points. The $1$ data point is the uniform superposition. The $n$ is being completely in one of the $n$ measurable states. The ${n \choose 2}$ is being in an evenly split superposition of two of the $n$ basis states. 

We generated the expected output data using the true Hamiltonian as a blackbox. 


\subsubsection*{Results}

\begin{table}
\centering
\begin{tabular}{ c | c | c | c}
\hline
\multicolumn{4}{|c|}{size of the Hamiltonian} \\
\hline
 & $4\times 4$ & $8\times 8$ & $30\times 30$ \\
\hline
average accuracy & 2.200358254 $\times$ $10^{-5}$ & 6.17835244 $\times$ $10^{-5}$ & 6.84412568 $\times$ $10^{-5}$ \\
\hline
\end{tabular}
\caption{\small Random Hamiltonian}
\end{table}

We generated 100 random Hamiltonians with $4$ quantum states, 100 random Hamiltonians with $8$ quantum states, and 10 random Hamiltonians with $30$ quantum states. Using the max norm of the difference between the random (true) Hamiltonian and the final learning Hamiltonian, we observed that each of the entries in these matrices was learned to at least 4 digits of accuracy. The cost function fell to at least 10e-10 for all matrices as well. In summary, the error fell to machine zero on all of our learning Hamiltonians.

%

This method can also be used online with a stochastic gradient descent method. One could either take a descent step with every run of the experiment, or batch them into small amounts. Numerical tests on the hyperfine splitting Hamiltonian found that the method would only converge if already in a small neighborhood of the correct answer. The issue with this approach is that sometimes a quantum state collapses to one of its low probability states, and this creates a large cost. Thus, it would be better to batch the data to reduce the frequency of collapses to low probability states. 


\subsubsection*{Brief note on efficiency}

Calculating the cost function can be an expensive calculation to perform since it involves a matrix exponential. Matrix exponentials are still an open area of research. Scipy, Octave and Matlab use the Pade approximant to find the matrix exponential. Since we have used a finite difference method, the algorithm makes many calls to the cost function, so a faster way to find the matrix exponential would create a large speed up. A matrix, $A$, is said to be diagonalizable if it can be written as $A = U D U^{-1}$ where $D$ is a diagonal matrix. There is a fast and simple method to calculate the exponential of a diagonal matrix, which is

\begin{equation*}
    e^A = U e^D U^{-1}.
\end{equation*}

In quantum mechanics, it is well known that the Hamiltonian is a hermitian operator. Every hermitian operator is diagonalizable. However, the Sch\"odinger equation requires the exponential of $it$ multiplied by the Hamiltonian matrix, so the question becomes: Is $i t H$ a diagonalizable matrix? The answer is yes. More generally, and complex valued diagonalizable matrix multiplied by a complex number is again diagonalizable. Let $c$ be a complex number and $A$ be a diagonalizable matrix as before. Then $cA = U (cA) U^{-1}$. This also tells that $i t H$ has the same eigenvectors and as $H$ and the eigenvalues are just scaled by $it$.

We implemented this diagonalization matrix exponential technique and compared it to the Scipy matrix exponential.
\begin{figure}
\centering
\includegraphics[scale=0.5]{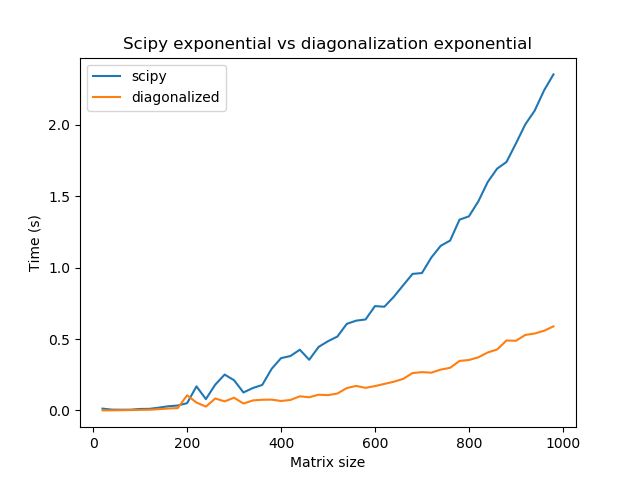}
\caption{\small We observed a large speed up using the diagonalization matrix exponential. Although unseen in this image, on matrices of size 4, we observed approximately a 100 times speed up.}
\end{figure}

Besides the obvious speed up observed, another advantage is that the Pade approximant is not the true matrix exponential. However, the difference we found between the diagonalization matrix exponential and the Pade approximant was machine epsilon using the max norm. 

\subsection{Discussion}

As we discuss before, DDA is a robust algorithm, i.e., it does not need any \emph{\'a priori} knowledge. However, it does not mean one has to ignore more tradition studies of quantum systems. Information known about the system can be incorporated to improve the accuracy of the method and reduce the run time. 

Consider the hyperfine splitting of hydrogen. By the conservation of total spin angular momentum one can argue that the Hamiltonian must be of the form

\[
\begin{bmatrix}
    w_{11}  & 0 & 0 & 0 \\
    0 & w_{22} & w_{23} & 0 \\
    0  & w_{23} & w_{33} & 0 \\
    0  & 0 & 0  & w_{44} \\
\end{bmatrix},
\]
where each $w_{ij}$ is a real number, and using the same basis as described earlier for this problem. 

When this form of the Hamiltonian is assumed, the algorithm tends to find a local minimum in around 9\% faster. More importantly, this information increases the probability of the initialization of the random weights placing the learning Hamiltonian in a convex neighborhood of the true Hamiltonian. Being in a convex neighborhood of the minimum essentially guarantees the success of the algorithm. 

Though the method has demonstrated that it can learn the Hamiltonian matrix quite successfully on a large variety of quantum systems, it does has some limitations. The first of which is that we have currently restricted ourselves to real-valued Hamiltonians, instead of the more general complex valued case. Our comments concerning this are twofold. First, one could define $n(n-1)/2$ additional weights, so that each off diagonal entry of the Hamiltonian has two weights. One weight for the real part of the entry and one weight for the imaginary part of the entry. Second, one could use the method in its current state, if one is content to only have the magnitude of each entry in the Hamiltonian matrix.

Another limitation of the method is that the true Hamiltonian will not necessarily be the global minimum depending on how much noise is in the experimental data. This was demonstrated in the case of the hyperfine splitting. However the minimum were very close. To this limitation, we note that one cannot expect one's results to be better than one's data. 

The final drawback we mention of the DDA is that the method only learns the numerical values. If one changes the experiment a small amount in the setting such as numbers of parameters and dimension of the input data set, one does not know anything about whats new, namely, we do not obtain a general and closed formula of the Hamiltonian matrix. What we can do is to repeatedly execute the routine to obtain the specific Hamiltonian of the specific setting.

One place where we suspect this method will be of great use to researchers is in the study of perturbation problems. One could use the known solution of a problem as the initial weights and allow the DDA to discover the Hamiltonian of the unknown solution of the perturbed problem. 

Concerning all of the limitations of the method, we believe that small changes to the methodology outlined in this paper could allow any of these limitations to be overcome. 


\section{Conclusion} \label{sec:conclusion}

We have presented a new data driven approach to studying quantum mechanical systems, in contrast to analytic methods. We have demonstrated the utility of the method on true quantum data as well as simulated data.  We have also allowed for the blending of analytic methods and this new machine learning approach. We believe this paper will serve as part of a foundation for new machine learning techniques for studying physical science problems.

\appendix 
\section{Minimum Hamiltonian Is Not Unique} \label{app:min_H_is_not_uniq}
\begin{thm} \label{thm:Hamiltonaina_is_not_uniq}
Let $C(H_L) = \sum_{i=1}^m 1 - |\bra{\phi_i} e^{-i\hbar H_L t_i} \ket{\psi_i}|^2$ be the cost function, cf, (\ref{eq:cost_func}). Then,
$$
C(H_L) = C(H_L + f\cdot I),
$$
where $f\in\reR$ and $I$ is the identity matrix.
\end{thm}
\begin{proof}
Let $H$ be a complex-valued hermitian matrix. Since $H$ is hermitian it
can be written in the form
\[
H=PAP^{-1}
\]
 for some nilpotent $P$ and diagonal $A=(a_i)$. Now consider the exponential
\[
e^{-iHt}
\]
 where $-i$ are scalar matrices and $t$ is real. In this case $t$ can be
absorbed into $H$ with it still being hermitian. So we have
\[
e^{-iHt}=e^{-iH}=e^{-i(PAP^{-1})}.
\]
Expand the exponential function, then we obtain
\begin{align*}
e^{-i(PAP^{-1})}&=\sum_{n=0}^\infty \frac{1}{n!}(-iPAP^{-1})^{n} 
=\sum_{n=0}^\infty \frac{1}{n!}(-i)^n PA^n P^{-1} \\
&=P \Big(\sum_{n=0}^\infty \frac{1}{n!}(-i)^n A^n \Big) P^{-1}
=Pe^{-iA} P^{-1}.
\end{align*}
Since $-iA$ is diagonal,
\[
\exp\Bigg(
\begin{bmatrix}
a_1 &  &  & \\
 & a_2 &  & \\
 &  & \ddots & \\
 &  &  & a_n
\end{bmatrix} \Bigg)=
\begin{bmatrix}
e^{a_1} &  &  & \\
 & e^{a_2} &  & \\
 &  & \ddots & \\
 &  &  & e^{a_n}
\end{bmatrix}
\]

Now consider an arbitrary scalar matrix
\[
g=
\begin{bmatrix}
f &  &  & \\
 & f &  & \\
 &  & \ddots & \\
 &  &  & f
\end{bmatrix},
\]
where $f\in\reR$. Since g commutes with any other matrix
\[
e^{g+A}=e^{g}e^{A}.
\]
In the context of our time evolution, this means that
\[
Pe^{-iA}P^{-1}=e^{-iH}=Pe^{-i(g+(A-g))}P^{-1}=Pe^{-ig}e^{-i(A-g)}P^{-1}.
\]
Now consider the probability distribution after a measurement of
this state:
\begin{align*}
|\bra{x}Pe^{-iA}P^{-1}\ket{\psi}|^{2}&=\bra{x}Pe^{-iA}P^{-1}\ket{\psi} 
\overline{\bra{x}Pe^{-iA}P^{-1}\ket{\psi}} \\
&=\bra{x}Pe^{-iA}P^{-1}\ket{\psi} \bra{\psi}Pe^{iA}P^{-1}\ket{x}.
\end{align*}
Since $e^{g}$ is a scalar matrix, it commutes with all the
other matrices. Then, we have
\begin{align*}
&|\bra{x}Pe^{-i(A-g)}P^{-1}\ket{\psi}|^{2} 
= \bra{x}Pe^{ig}e^{-iA}P^{-1}\ket{\psi} \bra{\psi}Pe^{-ig}e^{iA}P^{-1}\ket{x} \\
&=\bra{x}Pe^{-iA}P^{-1} \,\underbrace{(e^{ig}\ket{\psi} \bra{\psi}e^{-ig})}_{=\ket{\psi} \bra{\psi}}\,
Pe^{iA}P^{-1}\ket{x} \\
&=\bra{x}Pe^{-iA}P^{-1}\ket{\psi} \bra{\psi}Pe^{iA}P^{-1}\ket{x} = |\bra{x}Pe^{-iA}P^{-1}\ket{\psi}|^{2}.
\end{align*}
 Therefore the two probability distributions
\[
|\bra{x}Pe^{-i(A-g)}P^{-1}\ket{\psi}|^{2} = |\bra{x}Pe^{-iA}P^{-1}\ket{\psi}|^{2}
\]
are identical. In other words, there are infinitely many Hamiltonians that can minimize the cost function. 
\end{proof}


\begin{centering}
\subsection*{Acknowledgements}
\end{centering}
The authors, Jordan Burns and David Maughan, would like to thank professor Charles Torre and professor Andreas Malmendier for allowing and encouraging what began as a senior thesis to grow into a full fledged research project. Their mentorship has been invaluable in the production of this manuscript. 





\end{document}